\documentclass[11pt]{article}
\usepackage[a4paper,margin=20mm]{geometry}

\usepackage{authblk}
\title{\bf
Asymptotic Convertibility of Entanglement:\\
A General Approach to Entanglement Concentration and Dilution
}
\author[1]{Yong Jiao\thanks{shouyu@quest.is.uec.ac.jp}}
\author[2]{Eyuri Wakakuwa\thanks{wakakuwa@quest.is.uec.ac.jp}}
\author[2]{Tomohiro Ogawa\thanks{ogawa@is.uec.ac.jp}}
\affil[1]{
Graduate School of Information Systems, University of Electro-Communications,
1-5-1 Chofugaoka, Chofu-shi, Tokyo, 182-8585, Japan.
}
\affil[2]{
Graduate School of Informatics and Engineering, University of Electro-Communications,
1-5-1 Chofugaoka, Chofu-shi, Tokyo, 182-8585, Japan.
}
\date{\today}

\usepackage{color}
\usepackage{cite}
\usepackage{amsmath,amssymb}
\usepackage{braket}
\usepackage{amsthm}
\newtheorem{theorem}{Theorem}
\newtheorem{lemma}{Lemma}
\newtheorem{definition}{Definition}
\newtheorem{corollary}{Corollary}
\newtheorem{proposition}{Proposition}
\newtheorem{remark}{Remark}

\def\ve{\varepsilon}

\def\a{\alpha}

\def\r{\rho}
\def\s{\sigma}

\def\g{\gamma}
\def\N{\mathbb{N}}
\def\R{\mathbb{R}}

\def\X{\mathbb{X}}
\def\Y{\mathbb{Y}}
\def\x{\mathcal{X}}
\def\y{\mathcal{Y}}

\def\H{\mathcal{H}}

\def\ri{{n\rightarrow\infty}}

\def\uh{\underline H}
\def\oh{\overline H}
\def\ud{\underline D}
\def\od{\overline D}
\def\uc{\underline C}
\def\oc{\overline C}

\def\whr{\widehat\rho}
\def\whs{\widehat\sigma}

\def\whi{\widehat I}
\def\Tr{\mathrm{Tr}}
\def\whp{\widehat\psi}

\def\whP{\widehat\Phi}
\def\whph{\widehat\phi}
\def\bn#1{\{ #1 \}}
\def\ni{_{n=1}^{\infty}}

\def\nn{\nonumber \\}
\def\L{\mathcal{L}}
\def\F{\mathcal{F}}
\def\D{\mathcal{D}}

\def\pure#1{\ket{#1}\!\!\bra{#1}}
\def\norm#1{\Vert #1 \Vert}

\newenvironment{textmath}{\(\displaystyle}{\)}
\begin{document}
\maketitle

\abstract{
We consider asymptotic convertibility of an arbitrary sequence of bipartite pure states
into another by local operations and classical communication (LOCC).
We adopt an information-spectrum approach to address cases
where each element of the sequences is not necessarily in tensor power of
a bipartite pure state.  We derive necessary and sufficient conditions for the
LOCC convertibility of one sequence to another in terms of spectral entropy
rates of entanglement of the sequences.  Based on these results, we also provide
simple proofs for previously known results on the optimal rates of entanglement
concentration and dilution of general sequences of pure states.
}


\maketitle

\section{Introduction}

An entangled quantum state shared between two distant parties is used as a
resource for performing nonlocal quantum information processing. When a state is
not in a desired form as a resource, we may need to transform it by local
operations and classical communication (LOCC) to a state with the desired
form. Well-known examples of such tasks are {\em entanglement concentration and
dilution} \cite{BB96}. Entanglement concentration is a task to obtain a maximally
entangled state from many copies of a non-maximally entangled state by LOCC, and
entanglement dilution is its inverse process. When the initial state is copies
of a bipartite pure state, the optimal rates of entanglement concentration and
dilution are asymptotically equal to the entanglement entropy \cite{BB96}.

For cases where the initial and target states are not necessarily in tensor
power of a bipartite state, the {\em information-spectrum method} has been applied to
analyze entanglement concentration \cite{H06,BD08} and entanglement dilution
\cite{BD08}. Originally, the information-spectrum method was developed in
classical information theory by Han and Verd\'u \cite{HV93,VH94,H02} and has
been extended to quantum information theory by Nagaoka and Hayashi
\cite{N98,HN03,NH07}. In the setting of the information-spectrum method, the
optimal rates of entanglement concentration and dilution are obtained in terms of
{\em spectral entropy rates} \cite{H06,BD08}.

In this paper, we consider a more general situation in which an arbitrary sequence of bipartite pure states
$\widehat\psi^{AB}=\{\psi_n^{AB}\}_{n=1}^{\infty}$ is converted into another
$\widehat\phi^{AB}=\{\phi_n^{AB}\}_{n=1}^{\infty}$
asymptotically by a sequence of LOCC protocols $\widehat\L=\{{\mathcal L}_n\}_{n=1}^\infty$.
We require that the trace distance between the final state $\L_n(\psi_n^{AB})$
and the target state $\phi_n^{AB}$ vanishes in the limit of $n\rightarrow\infty$.
We address conditions in which such a conversion is possible.
Contrary to the previous approaches \cite{H06,BD08},
we do not assume that the target state or the initial state is a maximally entangled state.

The main results of this paper are as follows.
As a direct part of the convertibility, it is proved that the initial sequence $\whp^{AB}$
can be converted to the target $\whph^{AB}$ asymptotically
if the inf-spectral entropy of entanglement of $\whp^{AB}$
is larger than the sup-spectral entropy of entanglement of $\whph^{AB}$.
As a converse part, we prove that if $\whp^{AB}$ is convertible to $\whph^{AB}$,
the inf-/sup-spectral entropy of entanglement of $\whp^{AB}$ must be larger than those of $\whph^{AB}$, respectively.
If we restrict $\whph^{AB}$ or $\whp^{AB}$ to be a sequence of maximally entangle states,
our results turn out to be those obtained by Hayashi \cite{H06} and Bowen-Datta \cite{BD08},
regarding the optimal rates of entanglement concentration and dilution.

Our proof of the direct part is based on the theory of classical random number generation
and much simpler than those of \cite{BD08,H06}.
It has been pointed out by Kumagai and Hayashi \cite{KH13NF} that
there are close relations between convertibility of entanglement and classical random number generation
mainly on the second order analysis of convertibility of entanglement in the i.i.d.~setting.
In this paper, we pursue generality of such an idea in the information-spectrum setting
and provide a simple argument for the asymptotic convertibility of entanglement.

This paper is organized as follows. In section 2, we provide definitions of the problem and state the main results.
Proofs of the main results are presented in Section 3 and Section 4. Conclusions are given in section 5.

\section{Main Results}

In this section, we present definitions of the problem and state the main results of this paper.
Let $\H_n^A$ and $\H_n^B$ $(n=1,2,\dots)$ be arbitrary finite-dimensional Hilbert spaces
and consider a general sequence of bipartite systems $\H_n^{AB}=\H_n^A\otimes\H_n^B$ $(n=1,2,\dots)$ composed of them.
Let $\ket{\psi_n^{AB}}$ and $\ket{\phi_n^{AB}}$ in $\H_n^{AB}$ be arbitrary pure states for each $n\in\N$.
For simplicity of the notation, we denote density operators by
$\psi_n^{AB}=\pure{\psi_n^{AB}}$ and $\phi_n^{AB}=\pure{\phi_n^{AB}}$.
For arbitrary density operators $\psi_n^{AB}$,
the reduced density operators are written as
$\psi_n^{A}={\rm Tr}_B[\psi_n^{AB}]$ and $\psi_n^{B}={\rm Tr}_A[\psi_n^{AB}]$.

\subsection{Necessary and Sufficient Conditions for LOCC Convertibility}

For arbitrary sequences of bipartite pure states
$\widehat \psi^{AB}=\{\psi_n^{AB}\}_{n=1}^{\infty}$ and $\widehat \phi^{AB} = \{\phi_n^{AB}\}_{n=1}^{\infty}$,
we seek for conditions under which $\psi_n^{AB}$ can be converted into $\phi_n^{AB}$
by LOCC for each $n$, up to a certain error that vanishes in the limit of $n\rightarrow\infty$.

\begin{definition}\label{ac}
We say that $\widehat\psi^{AB}$ can be asymptotically converted into $\widehat\phi^{AB}$ by LOCC,
if there exists a sequence of LOCC $\L_n$ $(n=1,2,\dots)$ such that
\begin{align}
\lim_{\ri}\norm{\L_n(\psi_n^{AB})-\phi_n^{AB}}_1=0,
\label{eq:001}
\end{align}
where $\norm{\cdot}_1$ is the trace norm defined by $\norm{A}_1={\rm Tr}|A|$ for a operator $A$.

\end{definition}
In this paper, we provide necessary and sufficient conditions for the asymptotic convertibility
of two sequences of bipartite pure states in terms of {\it spectral entropy rates},
which are key ingredients in information-spectrum methods and defined as follows.
Let $\whr=\{\rho_n\}_{n=1}^{\infty}$ be an arbitrary sequence of density operators and
$\whs=\{\sigma_n\}_{n=1}^{\infty}$ be an arbitrary sequence of Hermitian operators.
Then, for each $\ve\in[0,1]$, the spectral divergence rates \cite{NH07} are defined by
\begin{align}
\ud(\varepsilon|\whr||\whs)
&:=\sup\left\{a\bigm| \liminf_{\ri} \Tr \rho_n\{\rho_n-e^{na}\sigma_n>0\}\ge1-\ve\right\},
\label{eq:002} \\
\od(\ve|\whr||\whs)
\label{eq:003}
&:=\inf\left\{a\bigm| \limsup_{\ri} \Tr \rho_n\{\rho_n-e^{na}\sigma_n>0\}\le\ve\right\}.
\end{align}
Here, $\{A>0\}$ denotes the spectral projection corresponding to the positive part of a Hermitian operator $A$.
Specifically, using the spectral decomposition $A=\sum_ka_kE_k$, $\{A>0\}$ is defined by
\begin{align*}
\{A>0\}=\sum_{k:\,a_k>0} E_k.
\end{align*}
With the spectral divergence rates, the spectral entropy rates \cite{NH07,BD06} are defined by
\begin{align}
\label{eq:005}
\uh(\ve|\whr):=-\od(\ve|\whr||\whi), \quad
\oh(\ve|\whr):=-\ud(\ve|\whr||\whi)
\end{align}
for $\ve\in[0,1]$, where $\whi=\{I_n\}_{n=1}^{\infty}$ is the sequence of identity operators.
Especially, for $\ve=0$ we write
\begin{align*}
\uh(\whr)=\uh(0|\whr), \quad
\oh(\whr)=\oh(0|\whr).
\end{align*}

For any general sequences of bipartite pure states $\widehat \psi^{AB}$, consider sequences of reduced states $\widehat\psi^A=\{\psi_n^{A}\}_{n=1}^{\infty}$ and $\widehat\psi^B=\{\psi_n^{B}\}_{n=1}^{\infty}$. Then it is clear that $\widehat\psi^A$ and $\widehat\psi^B$ have the same entropy spectral rates, i.e.,
\begin{align*}
\uh(\widehat\psi^A)=\uh(\widehat\psi^B), \quad
\oh(\widehat\psi^A)=\oh(\widehat\psi^B).
\end{align*}
The main results of this paper are as follows. 

\begin{theorem}[direct part]\label{thm:direct}
Let $\widehat\psi^{AB}=\{\psi_n^{AB}\}_{n=1}^{\infty}$ and $\widehat\phi^{AB}=\{\phi_n^{AB}\}_{n=1}^{\infty}$
be general sequences of bipartite pure states on $\H_n^{AB}$ $(n=1,2,\dots)$. If $\uh(\widehat\psi^A)>\oh(\widehat\phi^A)$ holds, then $\widehat\psi^{AB}$ can be asymptotically converted into $\widehat\phi^{AB}$ by LOCC.
\end{theorem}

\begin{theorem}[converse part]\label{thm:converse}
Let $\widehat\psi^{AB}=\{\psi_n^{AB}\}_{n=1}^{\infty}$ and $\widehat\phi^{AB}=\{\phi_n^{AB}\}_{n=1}^{\infty}$
be general sequences of bipartite pure states on $\H_n^{AB}$ $(n=1,2,\dots)$.
If $\widehat\psi^{AB}$ can be asymptotically converted into $\widehat\phi^{AB}$ by LOCC, it must hold that $\oh(\ve|\widehat\psi^A)\ge\oh(\ve|\widehat\phi^A)$ and $\uh(\ve|\widehat\psi^A)\ge\uh(\ve|\widehat\phi^A)$ for every $\ve\in [0,1]$.
\end{theorem}

\subsection{Entanglement concentration and dilution}
In this section, we use the above theorems to provide simple proofs of known results on the optimal rates of entanglement concentration and entanglement dilution for general sequences of bipartite pure states.

Let $\{M_n\}_{n=1}^\infty$ be an arbitrary sequence of natural numbers, and  let $|\Phi_{M_n}\rangle\in\H^{AB}_n$ be a maximally entangled state with Schmidt rank $M_n$ for each $n$. As a shorthand notation, we denote $\Phi_{M_n}^{AB}=\pure{\Phi_{M_n}}$. Noting that $\Phi_{M_n}^A=\Tr_B[\Phi_{M_n}^{AB}]$ and $\Phi_{M_n}^B=\Tr_A[\Phi_{M_n}^{AB}]$ are the maximally mixed states with Schmidt rank $M_n$, it is straightforward to verify that
\begin{align}
\uh(\widehat\Phi^A)=\liminf_{n\rightarrow\infty}\frac{1}{n}\log{M_n}, \quad
\oh(\widehat\Phi^A)=\limsup_{n\rightarrow\infty}\frac{1}{n}\log{M_n} \label{eq:uhohPhi}
\end{align}
for $\widehat \Phi^{A}=\{\Phi_{M_n}^{A}\}_{n=1}^{\infty}$.

\subsubsection{Entanglement Concentration}

Entanglement concentration is a task for two distant parties to obtain a sequence of maximally entangled states $\widehat \Phi^{AB}$
from a sequence of bipartite pure states $\whp^{AB}$ by LOCC.
\begin{definition}[Entanglement concentration rate] \label{de}
For a sequence $\whp^{AB}=\{\psi_n^{AB}\}_{n=1}^{\infty}$,
a rate $R$ is said to be achievable if there exists a sequence of natural numbers $\{M_n\}_{n=1}^\infty$
such that $\whp^{AB}$ can be asymptotically converted into $\widehat{\Phi}^{AB}=\{\Phi_{M_n}^{AB}\}_{n=1}^{\infty}$ by LOCC and
\begin{align*}
\liminf_{\ri}\frac{1}{n}\log M_n\ge R
\end{align*}
holds. The entanglement concentration rate, or distillable entanglement \cite{BD08}, of a sequence $\whp^{AB}$ is defined by
\begin{align}
R(\whp^{AB}):=\sup\Set{R|\text{$R$ is achievable}}.
\end{align}
\end{definition}

\begin{proposition}[{Hayashi \cite[Theorem 1]{H06}}, {Bowen-Datta \cite[Theorem 3]{BD08}}]\label{ha}
For a sequence of bipartite pure states $\whp^{AB}=\{\psi_n^{AB}\}_{n=1}^{\infty}$, we have
\begin{align}
R(\whp^{AB})=\uh(\whp^A).\label{eq:distillation}
\end{align}
\end{proposition}

\begin{proof}
We apply Theorem \ref{thm:direct} and Theorem \ref{thm:converse} regarding the target state $|\phi_n^{AB}\rangle$ as $\ket{\Phi_{M_n}}$.
From Theorem \ref{thm:direct} and \eqref{eq:uhohPhi}, $\whp^{AB}$ can be asymptotically converted into $\whP^{AB}$ by LOCC
if $M_n=e^{nR}$ and $\uh(\whp^A)>\oh(\whP^A)=R$. Thus a rate $R$ is achievable if $\uh(\whp^A)>R$.
Conversely, suppose that a rate $R$ is achievable.
By Definition \ref{de}, there exists a sequence $\whP^{AB}=\{\Phi_{M_n}^{AB}\}_{n=1}^\infty$ such that $\widehat\psi^{AB}$ can be asymptotically converted into $\whP^{AB}$
and \begin{textmath}\liminf_{\ri}\frac{1}{n}\log M_n\ge R \end{textmath}.
Then from Theorem \ref{thm:converse} and \eqref{eq:uhohPhi}, it must hold that
\begin{align*}
\uh(\whp^A)\ge\uh(\whP^A)=\liminf_{\ri}\frac{1}{n}\log M_n\ge R.
\end{align*}
Thus we obtain \eqref{eq:distillation}.
\end{proof}


\subsubsection{Entanglement Dilution}
Entanglement dilution is a  task for two distant parties to convert a sequence of maximally entangled states $\whP^{AB}$ into a sequence of bipartite pure states
$\whph^{AB}$ asymptotically by LOCC. 
\begin{definition}[Entanglement dilution rate] \label{ed}For a sequence $\whph^{AB}=\{\phi_n^{AB}\}_{n=1}^{\infty}$,
a rate $R$ is said to be achievable if there exists a sequence of natural numbers $\{M_n\}_{n=1}^\infty$
such that $\widehat{\Phi}^{AB}=\{\Phi_{M_n}^{AB}\}_{n=1}^{\infty}$ can be asymptotically converted into $\whph^{AB}$ by LOCC and
\begin{align*}
\limsup_{\ri}\frac{1}{n}\log M_n\le R
\end{align*}
holds. The entanglement dilution rate, or entanglement cost \cite{BD08}, of a sequence $\whph^{AB}$ is defined by
\begin{align}
R^*(\whph^{AB}):=\inf\Set{R|\text{$R$ is achievable}}.
\end{align}
\end{definition}

\begin{proposition}[{Bowen-Datta \cite[Theorem 4 ] {BD08}}]\label{bd}
For a sequence of bipartite pure states $\whph^{AB}=\{\phi_n^{AB}\}_{n=1}^{\infty}$, we have
\begin{align}
R^*(\whph^{AB})=\oh(\whph^A).\label{eq:dilution}
\end{align}
\end{proposition}

\begin{proof}
We apply Theorem \ref{thm:direct} and Theorem \ref{thm:converse} by taking the initial state $|\psi_n^{AB}\rangle$ as $\ket{\Phi_{M_n}}$.
From Theorem \ref{thm:direct} and \eqref{eq:uhohPhi}, $\whP^{AB}$ can be asymptotically converted into $\whph^{AB}$
if $M_n=e^{nR}$ and $R=\uh(\whP^A)>\oh(\whph^A)$. Thus a rate $R$ is achievable if $R>\oh(\whph^A)$.
Conversely, suppose that a rate $R$ is achievable. By Definition \ref{ed},
there exists a sequence $\whP^{AB}=\{\Phi_{M_n}\}_{n=1}^\infty$ such that $\whP^{AB}$ can be asymptotically converted into $\widehat\phi^{AB}$
and \begin{textmath}\limsup_{\ri}\frac{1}{n}\log M_n\le R\end{textmath}.
From Theorem \ref{thm:converse} and \eqref{eq:uhohPhi}, it must hold that
\begin{align*}
R\ge\limsup_{\ri}\frac{1}{n}\log M_n=\oh(\whP^A)\geq\oh(\whph^A).
\end{align*}
Thus we obtain (\ref{eq:dilution}).
\end{proof}


\section{Direct Part}
\label{sec:direct}

In this section, we give a proof of Theorem~\ref{thm:direct} using known results on classical random number generations.

\subsection{Random Number Generation and Majorization}

Let us first review the information-spectrum approach for random number generation \cite{H02}, introducing the spectral entropy rates of classical random variables.
For an arbitrary sequence of real valued random variables $\{Z_n\}_{n=1}^{\infty}$, we define the {\em limit superior and inferior in probability} by
\begin{align*}
\text{p\,-}\limsup_{\ri}Z_n := \inf\Set{ \a | \lim_{\ri} \Pr\{Z_n>\a\}=0 },\\
\text{p\,-}\liminf_{\ri}Z_n := \sup\Set{ \a | \lim_{\ri} \Pr\{Z_n<\a\}=0 }.
\end{align*}
Let $\X=\{X^n\}_{n=1}^{\infty}$ an arbitrary sequence of random variables, called a general source,
taking values in arbitrary countable sets ${\mathcal X}^n$ $(n=1,2,\dots)$,
and $P_{X^n}(x^n)$ $(x^n\in\x^n)$ be the probability function of $X_n$ for each $n$.
Then the {\em spectral entropy rates} of $\X$ is defined by
\begin{align}
\label{eq:uhx}
\underline{H}(\X):=\text{p\,-}\liminf_{\ri}\frac{1}{n}\log\frac{1}{P_{X^n}(X^n)},\quad
\overline{H}(\X):=\text{p\,-}\limsup_{\ri}\frac{1}{n}\log\frac{1}{P_{X^n}(X^n)}.
\end{align}

Let $Y$ and ${\tilde Y}$ be random valuables on a countable set $\mathcal Y$
and let $q(y)$ and ${\tilde q}(y)$ $(y\in{\mathcal Y})$ be the corresponding probability functions, respectively.
Then the {\em variational distance} between $Y$ and ${\tilde Y}$ is defined by
\begin{align}
d(Y,{\tilde Y}):=\sum_{y\in{\mathcal Y}}|q(y)-{\tilde q}(y)|.
\end{align}

\begin{proposition}[{Nagaoka \cite[Theorem 2.1.1] {H02}}]\label{prop:Nag}
Let $\X=\{X^n\}_{n=1}^{\infty}$ and $\Y=\{Y^n\}_{n=1}^{\infty}$ be arbitrary general sources. If $\overline{H}(\Y)<\underline{H}(\X)$,
then there exists a sequence of maps $\varphi_n: \x^n\rightarrow\y^n$ $(n=1,2,\dots)$
such that
\begin{align*}
\lim_{\ri} d(Y^n, \varphi_n(X^n))=0.
\end{align*}
\end{proposition}

Next, we treat a relation between random number generation and majorization.
For a sequence $a=\{a_i\}_{i=1}^m$ $(m\in\N)$, let $a^{\downarrow}=\{a_i^{\downarrow}\}_{i=1}^m$
denotes the sequence rearranged in decreasing order.
We say $a=\{a_i\}_{i=1}^m$ is majorized by $b=\{b_i\}_{i=1}^m$ and write $a\prec b$ if we have
\begin{align*}
\sum_{i=1}^k a_i^{\downarrow} \le \sum_{i=1}^k b_i^{\downarrow} \quad (k=1,2,\dots,m)
\end{align*}
and the equality for $k=m$.
Note that the majorization relation $a\prec b$ can be defined even when the numbers of elements in $a$ and $b$ differs, by including zero if necessary.
When both $a \prec b$ and $b \prec a$ hold, or equivalently $a^{\downarrow}=b^{\downarrow}$, we wirte $a \thicksim b$.

The following fact is given by Kumagai-Hayashi \cite{KH13NF}.
We show a proof here for readers' convenience since we can not find a proof in the literature.
\begin{lemma}[{Kumagai-Hayashi\cite[Section 3.2]{KH13NF}}]\label{lem:KH}
Given a map $\varphi:\x\to\y$ from a finite set $\x$ to $\y$,
and a probability function $p:x\in\x \mapsto p(x)\in[0,1]$ on $\x$, let
\begin{align*}
q(y)=\sum_{x\in \varphi^{-1}(\{y\})} p(x)
\end{align*}
be the induced probability function on $\y$.
Then we have $p \prec q$.
\end{lemma}

\begin{proof}
For each $y\in\y$, let $n(y)=|\varphi^{-1}(\{y\})|$ and $\varphi^{-1}(\{y\})=\{x_{y,1} , x_{y,2} , \dots , x_{y,n(y)}\}$,
and define $n(y)$-dimensional real vectors by
\begin{align*}
&\alpha_y := \left( p(x_{y,1}) , p(x_{y,2}) , \dots , p(x_{y,n(y)}) \right)^t,\\
&\beta_y:=\left(q(y),0,\dots,0\right)^t,
\end{align*}
where $(\dots)^{t}$ denotes the transposition of the vector.
It is straightforward to verify that $\alpha_y\prec \beta_y$ holds.
Thus there exists a doubly stochastic matrix $\D_y$ such that $\alpha_y=\D_y \beta_y$.
Indeed, letting
\begin{align}
\D_y=\sum_{j=1}^{n(y)}\frac{p(x_{y,j})}{q(y)}U_{n(y),j}
\end{align}
gives the relation $\alpha_y=\D_y \beta_y$, where $U_{n(y),j}$ is a $n(y)$ dimensional permutation matrix transposing the $1$st and $j$-th elements.
Since $\D_y$ is a convex combination of permutation matrices, it is doubly stochastic.
Now let us introduce a notation for the direct sum of vectors $u\in\R^n$ and $v\in\R^m$,
and the corresponding direct sum of matrices $A\in\R^{n\times n}$ and $B\in\R^{m\times m}$, by
\begin{align*}
u\oplus v=\begin{pmatrix} u \\ v \end{pmatrix}, \quad
A\oplus B=\begin{pmatrix} A & 0 \\ 0 & B \end{pmatrix}.
\end{align*}
Then we have $p\thicksim\bigoplus_{y\in\y} \alpha_{y}$ and $q\thicksim \bigoplus_{y\in\y}\beta_{y}$,
and hence,
\begin{align*}
p \thicksim \bigoplus_{y\in\y} \alpha_{y} = \bigoplus_{y\in\y} D_{y} \beta_{y}
 = \biggl(\bigoplus_{y\in\y} D_{y} \biggr) \biggl(\bigoplus_{y\in\y} \beta_{y} \biggr)
\prec \bigoplus_{y\in\y} \beta_{y}
\thicksim q,
\end{align*}
where the majorization $\prec$ follows from the fact that $\bigoplus_{y\in\y} \D_{y}$ is a doubly stochastic matrix.
\end{proof}

We note that the above lemma and the proof are valid for countable sets $\x$ and $\y$.

\subsection{Proof of Theorem \ref{thm:direct}}

Let $\ket{\psi_n^{AB}}$ and $\ket{\phi_n^{AB}}$ $(n=1,2,\dots)$ be the initial and target states, respectively, and
\begin{align*}
\ket{\psi_n^{AB}}&=\sum_{x^n\in\x^n}\sqrt{p_n(x^n)}\ket{e_{x^n}^{A}}\otimes\ket{e_{x^n}^B}, \\
\ket{\phi_n^{AB}}&=\sum_{y^n\in\y^n}\sqrt{q_n(y^n)}\ket{f_{y^n}^A}\otimes\ket{f_{y^n}^B}
\end{align*}
be their Schmidt decompositions.
Then their reduced density operators are given by
\begin{align*}
\psi^A_n&=\Tr_B\left[\psi_n^{AB}\right]=\sum_{x^n\in\x^n} p_n(x^n)\pure{e_{x^n}},\\
\phi^A_n&=\Tr_B\left[\phi_n^{AB}\right]=\sum_{y^n\in\y^n} q_n(y^n)\pure{f_{y^n}}.
\end{align*}
From the Schmidt coefficients we can define random variables $X^n$ and $Y^n$
subject to probability functions $p_n(x^n)$ $(x^n\in\x^n)$ and $q_n(y^n)$ $(y^n\in\y^n)$,
and general sources $\X=\{X^n\}_{n=1}^{\infty}$ and $\Y=\{Y^n\}_{n=1}^{\infty}$ composed of them.
For sequences of density operators $\whp^A=\{\psi_n^A\}\ni$ and $\whph^A=\{\phi_n^A\}\ni$,
it is straightforward to verify that
\begin{align}
\underline{H}(\X)=\uh(\widehat\psi^A),\quad\overline{H}(\Y)=\oh(\widehat\phi^A).
\end{align}

Suppose that $\uh(\widehat\psi^A)>\oh(\widehat\phi^A)$, or equivalently $\uh(\X)>\oh(\Y)$.
From Proposition \ref{prop:Nag}, there exists a sequence of maps $\varphi_n: \x^n\rightarrow\y^n$ $(n=1,2,\dots)$
such that the variational distance between $\tilde q_n(y^n)= p(\varphi_n^{-1}(\{y^n\}))$ and $q_n(y^n)$  $(y^n\in\mathcal Y)$ goes to zero asymptotically, i.e.,
\begin{align}
\lim_{n\to\infty}d(Y^n,{\tilde Y}^n)=0,\label{eq:vardistyy}
\end{align}
where ${\tilde Y}^n$ is a random variable subject to the probability function $\tilde q_n(y^n)$.
From Lemma \ref{lem:KH} it implies that $p_n\prec {\tilde q}_n$.

Consider a state
\begin{align*}
\ket{{\tilde\phi}_n^{AB}}:=\sum_{y^n\in\y^n}\sqrt{\tilde q_n(y^n)}\ket{f_{y^n}^A}\otimes\ket{f_{y^n}^B}.
\end{align*}
Due to Nielsen's theorem \cite{N99}, $\ket{\psi_n^{AB}}$ can be deterministically converted to $\ket{\tilde\phi_n^{AB}}$ by LOCC for each $n$.

To complete the proof, we verify that the state $\ket{\tilde\phi_n^{AB}}$ is equal to the target state $\ket{\phi_n^{AB}}$ asymptotically. Let $F(\rho,\sigma)$ be the fidelity between state $\rho$ and $\sigma$, defined by $F(\rho,\sigma):=\Tr|\sqrt{\rho}\sqrt{\sigma}|$. Noting that
\begin{align*}
{\tilde \phi^A_n}&=\Tr_B\left[{\tilde\phi_n^{AB}}\right]=\sum_{y^n\in\y^n}\tilde q_n(y^n)\pure{f_{y^n}},
\end{align*}
we have
\begin{align}
F({\tilde\phi_n^{AB}},\phi_n^{AB})=|\langle \tilde\phi_n,\phi_n \rangle|=\sum_{y^n\in\y^n}\sqrt{\tilde q_n(y^n)q_n(y^n)}=F({\tilde\phi_n^{A}},\phi_n^{A}).\label{eq:FABeqFA}
\end{align}

It is well known \cite{NC00} that the trace distance and the fidelity are related as
\begin{align}
1-F(\rho,\sigma)\le \norm{\rho-\sigma}_1\le\sqrt{1-F(\rho,\sigma)^2}.\label{eq:fintd}
\end{align}
Noting $\norm{ {\tilde\phi_n^A}-\phi_n^A }_1=d(Y^n,{\tilde Y}^n)$, from (\ref{eq:vardistyy}) we have
\begin{align}
\lim_{n\rightarrow\infty}\norm{ {\tilde\phi_n^A}-\phi_n^A }_1=0,
\end{align}
which implies
\begin{align}
\label{eq:lr}
\lim_{n\rightarrow\infty}F({\tilde\phi_n^A},\phi_n^A )=1
\end{align}
from the first inequality of (\ref{eq:fintd}). From (\ref{eq:FABeqFA}), it implies
\begin{align}
\lim_{n\rightarrow\infty}F({\tilde\phi_n^{AB}},\phi_n^{AB})=1,
\end{align}
which leads to
\begin{align}
\lim_{\ri}\norm{\tilde\phi_n^{AB}-\phi_n^{AB}}_1 = 0
\end{align}
due to the second inequality of (\ref{eq:fintd}).\hfill$\square$

\section{Converse Part}
\label{sec:converse}

In this section, we prove Theorem \ref{thm:converse} after reviewing properties of spectral divergence rates.

\subsection{Properties of Spectral Divergence Rates}
It is proved in \cite{BD06} that spectral divergence rates have properties of monotonicity and continuity for $\ve=0$. We extend these results to any $\ve\in [0,1]$. We also prove an inequality for the spectral entropy rates of a sequence of product states.


\subsubsection{Prerequisites}

Let $A$ be a Hermitian operator,
and let $A=\sum_k a_k E_k$ be the spectral decomposition.
Then the positive and negative parts of $A$ are, respectively, defined by
\begin{equation*}
A_+:=\sum_{k:\,a_k>0}a_kE_k, \quad A_-:=\sum_{k:\,a_k\le 0}(-a_k)E_k.
\label{eq:1}
\end{equation*}
Following \cite{N98,NH07},
we denote the corresponding projections by
\begin{equation*}
\{A>0\}:=\sum_{k:\,a_k>0}E_k, \quad \{A\le 0\}:=\sum_{k:\,a_k\le 0}E_k.
\label{eq:2}
\end{equation*}
With the above notations, we have $A_+=A\{A>0\}$ and $A_-=-A\{A\le 0\}$.
Note that
\begin{eqnarray}
A=A_ + -A_-,\quad |A|=A_ + +A_-
\label{eq:3}
\end{eqnarray}
are, respectively, the Jordan decomposition
and the absolute value of the operator $A$.
The following lemma is essential in information-spectrum methods.

\begin{lemma}\label{lem:np}
For any $0\le T\le I$, we have
\begin{eqnarray}
\Tr A_+ = \Tr A\{A>0\} \ge \Tr AT,
\label{eq:4}
\end{eqnarray}
or equivalently,
\begin{eqnarray}
\Tr A_+ = \max_{T:\,0\le T\le I}\Tr AT.
\label{eq:5}
\end{eqnarray}
\end{lemma}

It is also useful to note the relation with the trace norm:
\begin{eqnarray*}
\Tr A_+ =  \frac{1}{2}\{ \Tr |A| + \Tr A \},
\ \Tr A_- =  \frac{1}{2}\{ \Tr |A| - \Tr A \},
\label{eq:53}
\end{eqnarray*}
which follows from \eqref{eq:3}.
Especially, if $\Tr A = 0$
\begin{align}
\Tr |A| &= 2\, \Tr A_+  = 2\, \Tr A_- ,
\label{eq:54} \\
\Tr |A| &= 2\max_{T:\,0\le T\le I}\Tr AT.
\label{eq:55}
\end{align}
It should also be noted that from
\begin{textmath}
\Tr (A-B)_+ = \Tr (A-B)\{A-B>0\} \ge 0
\end{textmath},
we have
\begin{align}
\Tr A\{A-B>0\}\ge \Tr B\{A-B>0\}.
\label{eq:6}
\end{align}
and it obviously holds that
\begin{align}
\Tr (A-B)_+ = \Tr (A-B) \{A-B>0\} \ge \Tr A \{A-B>0\}.
\label{eq:7}
\end{align}
The following lemma was poited out by Bowen-Datta \cite{BD08} for completely positive and trace preserving maps,
It should be noted that $\F$ is no need to be complete positive map.
\begin{lemma}\label{bdm}
Let $A$ and $B$ be Hermitian operators. For any trace preserving (TP) maps $\F$, we have
$\Tr A_+ \ge \Tr \F(A)_+$.
\end{lemma}

\subsubsection{Monotonicity}

There is an alternative expression for the spectral divergence rates introduced by Bowen-Datta\cite{BD06}. For each $\ve\in[0,1]$, let
\begin{align*}
\uc(\ve|\whr||\whs)&:=\sup\left\{a\bigm|\liminf_{\ri}{\rm Tr}(\r_n-e^{na}\s_n)_+\ge1-\ve\right\}, \\
\oc(\ve|\whr||\whs)&:=\inf\left\{a\bigm |\limsup_{\ri}{\rm Tr}(\r_n-e^{na}\s_n)_+\le\ve\right\}.
\end{align*}
It can be shown that these apparently different definitions yield the same quantities \cite{BD06}.
\begin{lemma}\label{bd}
For any $\ve\in[0,1]$, we have
\begin{eqnarray}
\label{eq:00}
\uc(\ve|\whr||\whs)=\ud(\ve|\whr||\whs),\\
\label{eq:01}
\oc(\ve|\whr||\whs)=\od(\ve|\whr||\whs).
\end{eqnarray}
\end{lemma}
The proof is given in Appendix.

We use this lemma to prove the monotonicity of spectral divergence rates as follows.
\begin{proposition} \label{prop:monotonicity}
For any sequence of TP maps $\widehat{\F}=\{\F_n\}_{n=1}^{\infty}$, the monotonicity of the spectral divergence rates hold, that is,
\begin{align}
&\ud(\ve|\whr||\whs)\ge\ud\big(\ve|\widehat{\F}(\whr)||\widehat{\F}(\whs)\big)\label{eq:monotod},\\
&\od(\ve|\whr||\whs)\ge\od\big(\ve|\widehat{\F}(\whr)||\widehat{\F}(\whs)\big)\label{eq:monotud},
\end{align}
for any $\ve\in [0,1].$
\end{proposition}
\begin{proof}
 For any $\g >0$, choose $a=\ud\big(\ve|\widehat{\F}(\whr)||\widehat{\F}(\whs)\big)-\g$. From Lemma \ref {bdm}, we have
\begin{align}
\label{eq:tf}
1-\ve\le\Tr \big(\F_n(\r_n)-e^{na}\F_n(\s_n)\big)_+\le\Tr(\r_n-e^{na}\s_n)_+.
\end{align}
Taking \begin{textmath}\liminf_{\ri} \end{textmath} of (\ref{eq:tf}), we have
\begin{align}
1-\ve\le\liminf_{\ri}\Tr \big(\F_n(\r_n)-e^{na}\F_n(\s_n)\big)_+\le\liminf_{\ri}\Tr(\r_n-e^{na}\s_n)_+.
\end{align}
From (\ref{eq:00}) and the definition of $\uc$, we obtain $a=\ud\big(\ve|\widehat{\F}(\whr)||\widehat{\F}(\whs)\big)-\g\le \ud(\ve|\whr||\whs)$ for all $\g >0$, which implies (\ref{eq:monotod}).

Similarly, if we choose $a=\ud(\whr||\whs)+\g$, and from Lemma \ref {bdm}, we have
\begin{align}
\label{eq:ve}
 \Tr \big(\F_n(\r_n)-e^{na}\F_n(\s_n)\big)_+\le\Tr(\r_n-e^{na}\s_n)_+\le\ve,
\end{align}
Thus, taking \begin{textmath}\limsup_{\ri} \end{textmath} of (\ref{eq:ve}), we have
\begin{align}
 \limsup_{\ri}\Tr \big(\F_n(\r_n)-e^{na}\F_n(\s_n)\big)_+\le\limsup_{\ri}\Tr(\r_n-e^{na}\s_n)_+\le\ve.
\end{align}
From (\ref{eq:01}) and the definition of $\oc$, we have $\od\big(\ve|\widehat{\F}(\whr)||\widehat{\F}(\whs)\big)\le a=\od(\ve|\whr||\whs)+\g$ for all $\g>0$, which leads to (\ref{eq:monotud}).
\end{proof}
From (\ref{eq:005}), we have the following Corollary.
\begin{corollary}\label{ru}For any sequence of unital $TP$ maps $\widehat{\F}=\{\F_n\}_{n=1}^{\infty}$, the following inequalities hold for any $\ve\in [0,1]:$
\begin{eqnarray}
\label{eq:19}
\oh(\ve|\whr)\le\oh\big(\ve|\widehat{\F}(\whr)\big),\\
\label{eq:20}
\uh(\ve|\whr)\le\uh\big(\ve|\widehat{\F}(\whr)\big).
\end{eqnarray}
\end{corollary}

\subsubsection{Continuity}

Spectral divergence rates are ``continuous'' with respect to the sequences of density operators in the first argument, that is, spectral divergence rates of two sequences coincide if the sequences are asymptotically equal.

\begin{lemma}\label{td}
Let $\whr=\{\r_n\}_{n=1}^{\infty}$ and $\whr'=\{\r'_n\}_{n=1}^{\infty}$
be sequences of density operators. If
\begin{eqnarray}
\lim_{n\to\infty}||{\r_n-\r'_n}||_1=0,
\end{eqnarray}
then
\begin{eqnarray}
\label{eq:11}
\ud(\ve|\whr||\whs)&=&\ud(\ve|\whr'||\whs),\\
\label{eq:12}
\od(\ve|\whr||\whs)&=&\od(\ve|\whr'||\whs)
\end{eqnarray}
hold for any $0\le\ve\le 1$ and
any sequence $\whs=\{\s_n\}_{n=1}^{\infty}$ of Hermitian operators.\\
\end{lemma}

\begin{proof}
From \eqref{eq:55}, we have
\begin{align*}
\norm{\rho_n-\rho'_n}_1
&=||{(\rho_n-e^{na}\sigma_n)-(\rho'_n-e^{na}\sigma_n)}||_1 \\
&=\Tr|{(\rho_n-e^{na}\sigma_n)-(\rho'_n-e^{na}\sigma_n)}|\\
&\ge 2 \Tr(\rho_n-e^{na}\sigma_n)\{\rho_n-e^{na}\sigma_n>0\}\\
&-2\Tr(\rho'_n-e^{na}\sigma_n)\{\rho_n-e^{na}\sigma_n>0\} \\
&\ge 2\Tr(\rho_n-e^{na}\sigma_n)_+ -2\Tr(\rho'_n-e^{na}\sigma_n)_+
\end{align*}
where the last inequality follows from \eqref{eq:4}. Hence
\begin{align}
\Tr(\rho'_n-e^{na}\sigma_n)_+ +\frac{1}{2}||{\rho_n-\rho'_n}||_1
\ge \Tr(\rho_n-e^{na}\sigma_n)_+.
\label{eq:10}
\end{align}
For any $\g>0$, let $a=\ud(\ve|\whr||\whs)-\g$.
Then from $\ud(\ve|\whr||\whs)=\uc(\ve|\whr||\whs)$,
we have
\begin{eqnarray}
\liminf_{n\to\infty}\Tr(\rho_n-e^{na}\sigma_n)_+ \ge 1-\ve.
\end{eqnarray}
Thus taking \begin{textmath}\liminf_{\ri}\end{textmath} of \eqref{eq:10} gives
\begin{align*}
\liminf_{n\to\infty}\Tr(\rho'_n-e^{na}\sigma_n)_+
&\ge \liminf_{n\to\infty}\Tr(\rho_n-e^{na}\sigma_n)_+ \\
&\ge 1-\ve,
\end{align*}
which implies $a=\ud(\ve|\whr||\whs)-\g\le\ud(\ve|\whr'||\whs)$.
Since $\g>0$ can be arbitrary, we have
$\ud(\ve|\whr||\whs)\le\ud(\ve|\whr'||\whs)$.
Interchanging the role of $\whr$ and $\whr'$, we have
the converse inequality $\ud(\ve|\whr||\whs)\ge\ud(\ve|\whr'||\whs)$.
Thus we have \eqref{eq:11}.
In the same way, we have \eqref{eq:12}.
\end{proof}
From \eqref{eq:005}, we have the following corollary.
\begin{corollary}\label{coro}
Let $\whr=\{\rho_n\}_{n=1}^{\infty}$ and $\whr'=\{\rho'_n\}_{n=1}^{\infty}$
be sequences of density operators.
If
\begin{align}
\lim_{n\to\infty}\norm{\rho_n-\rho'_n}_1=0,
\label{eq:47}
\end{align}
then
\begin{align}
\uh(\ve|\whr)&=\uh(\ve|\whr'),
\label{eq:48}
\\
\oh(\ve|\whr)&=\oh(\ve|\whr')
\label{eq:49}
\end{align}
hold for any $0\le\ve\le 1$.
\end{corollary}

\subsubsection{ Spectral Entropy of Product States}
\begin{lemma}For arbitrary sequences $\whr^A=\bn{\r_n^A}\ni$ and $\whs^B=\bn{\s_n^B}\ni$, the followings hold for any $\ve\in[0,1]:$
\begin{align}
\label{eq:oh}
\oh(\ve|\whr^A\otimes\whs^B)\ge\oh(\ve|\whr^A),\\
\label{eq:uh}
\uh(\ve|\whr^A\otimes\whs^B)\ge\uh(\ve|\whr^A).
\end{align}
\end{lemma}
\begin{proof}

From (\ref{eq:002}), (\ref{eq:003}) and (\ref{eq:005}), we have
\begin{align}
\label{eq:vuhl}
&\oh(\ve|\whr^A)=\inf\left\{a\bigm|\liminf_{\ri}\Tr\r_n^A\left\{\r_n^A>e^{-na}I_n\right\}\ge1-\ve\right\},\\
\label{eq:vohl}
&\uh(\ve|\whr^A)=\sup\left\{a\bigm|\limsup_{\ri}\Tr\r_n^A\left\{\r_n^A>e^{-na}I_n\right\}\le\ve\right\},\\
{\text and}\\
\label{eq:vuh}
&\oh(\ve|\whr^A\otimes\whs^B)=\inf\left\{a\bigm|\liminf_{\ri}\Tr(\r_n^A\otimes\s_n^B)\left\{\r_n^A\otimes\s_n^B>e^{-na}I_n\right\}\ge1-\ve\right\},\\
\label{eq:voh}
&\uh(\ve|\whr^A\otimes\whs^B)=\sup\left\{a\bigm|\limsup_{\ri}\Tr(\r_n^A\otimes\s_n^B)\left\{\r_n^A\otimes\s_n^B>e^{-na}I_n\right\}\le\ve\right\}.
\end{align}

Let
\begin{align}
&\r_n^A=\sum_k\l_{n,k}^A\pure{\phi_{n,k}^A},\\
&\s_n^B=\sum_l\l_{n,l}^B\pure{\phi_{n,l}^B}
\end{align}
be spectral decompositions of $\r_n^A$ and $\r_n^B$. Then
(\ref{eq:vuhl}), (\ref{eq:vohl}), (\ref{eq:vuh}) and (\ref{eq:voh}) can be rewritten as,
\begin{align}
\label{eq:sii}
&\oh(\ve|\whr^A)=\inf\bigg\{a\bigm|\liminf_{\ri}\sum_{\substack{k\\-\frac{1}{n}\log\l^A_{n,k}\le a}}{\l_{n,k}^A}\ge1-\ve\bigg\},\\
\label{eq:iis}
&\uh(\ve|\whr^A)=\sup\bigg\{a\bigm|\limsup_{\ri}\sum_{\substack{k\\-\frac{1}{n}\log\l_{n,k}^A \le a}}  \l_{n,k}^A\le\ve\bigg\},\\
\label{eq:si}
&\oh(\ve|\whr^A\otimes\whs^B)=\inf\bigg\{a\bigm|\liminf_{\ri}\sum_{\substack{k,l\\-\frac{1}{n}\log\l^A_{n,k}-\frac{1}{n}\log\l_{n,l}^B\le a}}{\l_{n,k}^A\l_{n,l}^B}\ge1-\ve\bigg\},\\
\label{eq:is}
&\uh(\ve|\whr^A\otimes\whs^B)=\sup\bigg\{a\bigm|\limsup_{\ri}\sum_{\substack{k,l\\-\frac{1}{n}\log\l_{n,k}^A-\frac{1}{n}\log\l_{n,l}^B \le a}}  \l_{n,k}^A\l_{n,l}^B\le\ve\bigg\}.
\end{align}

First, we prove (\ref{eq:oh}).
\begin{align}
\label{eq:110}
1-\ve\le\sum_{\substack{k,l\\-\frac{1}{n}\log\l_{n,k}^A-\frac{1}{n}\log\l_{n,l}^B\le a}} \l_{n,k}^A\l_{n,l}^B\le\sum_{\substack{k,l\\-\frac{1}{n}\log\l_{n,k}^A\le a}} {\l_{n,k}^A\l_{n,l}^B}=\sum_{\substack{ k\\-\frac{1}{n}\log\l_{n,k}^A\le a}}\l_{n,k}^A,
\end{align}
taking \begin{textmath}\liminf_{\ri}\end{textmath} of (\ref{eq:110}) gives,
\begin{align}
\label{eq:111}
1-\ve\le\liminf_{\ri}\sum_{\substack{k,l\\-\frac{1}{n}\log\l_{n,k}^A-\frac{1}{n}\log\l_{n,l}^B\le a}} \l_{n,k}^A\l_{n,l}^B\le\liminf_{\ri}\sum_{\substack{ k\\-\frac{1}{n}\log\l_{n,k}^A\le a}}\l_{n,k}^A,
\end{align}
which implies (\ref{eq:oh}) from (\ref{eq:sii}) and (\ref{eq:si}).

Next, we prove (\ref{eq:uh}).
\begin{align}
\label{eq:112}
\sum_{\substack{k,l\\-\frac{1}{n}\log\l_{n,k}^A-\frac{1}{n}\log\l_{n,l}^B\le a}} \l_{n,k}^A\l_{n,l}^B\le\sum_{\substack{k,l\\-\frac{1}{n}\log\l_{n,k}^A\le a}} {\l_{n,k}^A\l_{n,l}^B}=\sum_{\substack{ k\\-\frac{1}{n}\log\l_{n,k}^A\le a}}\l_{n,k}^A\le\ve,
\end{align}
taking \begin{textmath}\limsup_{\ri}\end{textmath} of (\ref{eq:111}) gives,
\begin{align}
\label{eq:113}
\limsup_{\ri}\sum_{\substack{k,l\\-\frac{1}{n}\log\l_{n,k}^A-\frac{1}{n}\log\l_{n,l}^B\le a}} \l_{n,k}^A\l_{n,l}^B\le\limsup_{\ri}\sum_{\substack{ k\\-\frac{1}{n}\log\l_{n,k}^A\le a}}\l_{n,k}^A\le\ve,
\end{align}
which implies (\ref{eq:uh}) from (\ref{eq:iis}) and (\ref{eq:is}).
\end{proof}
\begin{remark}
This Lemma is pointed out by Bowen-Datta for $\ve=0$(see Corollary 7 in \cite{BD06}).
\end{remark}

\subsection{Description of A LOCC Protocol}\label{sec:LOCC}
Let us consider Alice and Bob engage in a multi-round LOCC protocol.
Without loss of generality, we assume that an LOCC protocol starts with Alice's measurement and end with Alice's operation on her system. Due to the Naimark extension theorem (\!\cite{Naimark40}, see also Theorem 4.5 in \cite{Hayashitext}), such a protocol can in general be described as follows:
\begin{itemize}
\item[1.] Alice performs an isometry operation $V_\gamma:A\rightarrow AE_{A,\gamma}^1E_{A,\gamma}^2$.
\item[2.] Alice performs a projective measurement on $E_{A,\gamma}^1$, and obtains an outcome.
\item[3.] Alice communicates a classical message to Bob.
\item[4.] Bob performs an isometry operation $W_\gamma:B\rightarrow BE_{B,\gamma}^1E_{B,\gamma}^2$.
\item[5.] Bob performs a projective measurement on $E_{B,\gamma}^1$, and obtains an outcome.
\item[6.] Bob communicates a classical message to Alice.
\item[7.] Alice and Bob recursively apply 1$\sim$6 for $\gamma=1,\cdots,\Gamma$, where $\Gamma\in{\mathbb N}$ is the number of rounds of the protocol.
\item[8.] Alice performs an isometry operation $V^*:A\rightarrow AE_A^*$, where $E_A^*$ is an ancillary system.
\item[9.] Alice and Bob discard ancillary systems $E_{A,1}^2\cdots E_{A,\Gamma}^2E_A^*$ and $E_{B,1}^2\cdots E_{B,\Gamma}^2$, respectively.
\end{itemize}

An advantage of introducing such a description is that, if the initial state is pure, the whole state remains pure until the last step in which Alice and Bob discard ancillary systems (Step 9 above).

\subsection{Proof of Theorem ~\ref{thm:converse}}

Theorem~\ref{thm:converse} is proved as follows. Suppose $\whp^{AB}=\{\psi_n^{AB}\}_{n=1}^{\infty}$ can be converted into $\whph^{AB}=\{\phi_n^{AB}\}_{n=1}^{\infty}$ asymptotically by LOCC. By Definition \ref{ac}, there exists a sequence of LOCC $\L_n\:(n=1,2,\cdots)$ such that
\begin{align}
\lim_{\ri}\norm{\L_n(\psi_n^{AB})-\phi_n^{AB}}_1=0.\label{eq:convLnpsi}
\end{align}
From (\ref{eq:fintd}), it leads to
\begin{align}
\lim_{\ri}F(\L_n(\psi_n^{AB}),\phi_n^{AB})=1\label{eq:convLnpsi2}.
\end{align}
 For each $n$, let $\L_n'$ be an LOCC protocol corresponding to Step 1$\sim$8 of $\L_n$ (see Section \ref{sec:LOCC}), and denote ancillary systems $E_{A,1}^2\cdots E_{A,\Gamma}^2E_A^*$ and $E_{B,1}^2\cdots E_{B,\Gamma}^2$ simply by $E_A$ and $E_B$, respectively. Define a pure state $|\phi_n'\rangle^{ABE_AE_B}$ by
\begin{align}
\phi'^{ABE_AE_B}_n=\L_n'(\psi_n^{AB}).\label{eq:phipr}
\end{align}
The final state of the protocol is then given by
\begin{align}
\L_n(\psi_n^{AB})={\rm Tr}_{E_AE_B}[\L_n'(\psi_n^{AB})]={\rm Tr}_{E_AE_B}[\phi'^{ABE_AE_B}_n].\label{eq:dfnphipr}
\end{align}

Due to Uhlmann's theorem \cite{Uhlmann76}, (\ref{eq:convLnpsi2}) and (\ref{eq:dfnphipr}), there exists a sequence of pure states ${\widehat \xi}^{E_AE_B}=\{\xi_n^{E_AE_B}\}_{n=1}^{\infty}$ such that $\phi'^{ABE_AE_B}_n$ is asymptotically equal to $\phi^{AB}_n\otimes\xi_n^{E_AE_B}$, i.e.,
\begin{align}
\lim_{\ri}F(\phi_n'^{ABE_AE_B},\phi^{AB}_n\otimes\xi_n^{E_AE_B})=1,
\end{align}
which implies
\begin{align}
\lim_{\ri}\norm{\phi_n'^{ABE_AE_B}-\phi^{AB}_n\otimes\xi_n^{E_AE_B}}_1=0
\end{align}
from (\ref{eq:fintd}). Due to Corollary \ref{coro}, we have
\begin{eqnarray}
\label{eq:17}
\uh(\ve|{\widehat{\phi'}}^{AE_A})&=&\uh(\ve|\widehat{\phi}^A\otimes\widehat{\xi}^{E_A}),\\
\label{eq:18}
\oh(\ve|{\widehat{\phi'}}^{AE_A})&=&\oh(\ve|\widehat{\phi}^A\otimes\widehat{\xi}^{E_A})
\end{eqnarray}
with $\widehat{\phi'}^{AE_A}=\{\phi'^{AE_A}_n\}_{n=1}^{\infty}$, $\widehat{\phi}^A=\{{\phi_n^A}\}_{n=1}^{\infty}$ and $\widehat{\xi}^{E_A}=\{\xi_n^{E_A}\}_{n=1}^{\infty}$.

From (\ref{eq:phipr}) and Nielsen's theorem (\!\!\cite{N99}, see also proof of Theorem 12.15 in \cite{NC00}), for each $n$, there exists a unital CPTP map
on ${\mathcal H}_n^A$ that maps $\phi'^{AE_A}_n$ to $\psi_n^{A}$. Applying the monotonicity of spectral inf-/sup-entropy rates (Corollary \ref{ru}), we have
\begin{align}
\label{eq:71}
&\uh(\ve|{\widehat{\phi'}}^{AE_A})\le\uh(\ve|{\whp}^A),\\
\label{eq:81}
&\oh(\ve|{\widehat{\phi'}}^{AE_A})\le\oh(\ve|{\whp}^A).
\end{align}
From (\ref{eq:oh}) and (\ref{eq:uh}), we also have
\begin{align}
\label{eq:91}
&\uh(\ve|{\widehat{\phi}}^{A})\le\uh(\ve|\widehat{\phi}^A\otimes\widehat{\xi}^{E_A}),\\
\label{eq:92}
&\oh(\ve|{\widehat{\phi}}^{A})\le\oh(\ve|\widehat{\phi}^A\otimes\widehat{\xi}^{E_A}).
\end{align}
From all above, we obtain
\begin{align}
&\uh(\ve|{\widehat{\phi}}^{A})\le\uh(\ve|{\whp}^A),\\
&\oh(\ve|{\widehat{\phi}}^{A})\le\oh(\ve|{\whp}^A).
\end{align}
\hfill$\square$
\section{Conclusion}
We analyzed asymptotic LOCC convertibility of sequences of bipartite pure entangled states, and derived necessary and sufficient conditions for a sequence to be asymptotically convertible to another. Applying these results, we also provided simple proofs for the optimal rates of entanglement concentration and dilution in an information-spectrum setting.

\section*{Appendix}

In this appendix, we give a proof of lemma \ref{bd}.
\begin{proof}
It obviously holds that
\begin{align}
\Tr(\rho_n-e^{na}\sigma_n)_+
&=\Tr(\rho_n-e^{na}\sigma_n)\{\rho_n-e^{na}\sigma_n>0\} \nn
&\le\Tr\rho_n\{\rho_n-e^{na}\sigma_n>0\}.
\label{eq:29}
\end{align}
Let $\g>0$ be arbitrary and
$a=\uc(\ve|\whr||\whs)-\g$.
From the Definition of $\uc(\ve|\whr||\whs)$, we have
\begin{align*}
\liminf_{n\to\infty}\Tr(\rho_n-e^{na}\sigma_n)_+ \ge 1-\ve.
\label{eq:30}
\end{align*}
Thus, taking \begin{textmath}\liminf_{\ri}\end{textmath} in the both sides of \eqref{eq:29}, we have
\begin{align}
\liminf_{n\to\infty}\Tr\rho_n\{\rho_n-e^{na}\sigma_n>0\}
\ge\liminf_{n\to\infty}\Tr(\rho_n-e^{na}\sigma_n)_+ \ge 1-\ve,
\end{align}
which implies
\begin{align}
\ud(\ve|\whr||\whs) &\ge a=\uc(\ve|\whr||\whs)-\g.
\label{eq:32}
\end{align}
Since $\g>0$ can be arbitrary, we have
\begin{align}
\ud(\ve|\whr||\whs) &\ge\uc(\ve|\whr||\whs).
\label{eq:33}
\end{align}
We show the converse inequality.
For any real number $a$ and $b$, \eqref{eq:4} yields
\begin{align}
\Tr(\rho_n-e^{na}\sigma_n)_+
&\ge\Tr(\rho_n-e^{na}\sigma_n)\{\rho_n-e^{nb}\sigma_n>0\} \nn
&=\Tr\rho_n\{\rho_n-e^{nb}\sigma_n>0\}
-e^{na}\Tr\sigma_n\{\rho_n-e^{nb}\sigma_n>0\} \nn
&\ge\Tr\rho_n\{\rho_n-e^{nb}\sigma_n>0\}
-e^{na}e^{-nb}\Tr\rho_n\{\rho_n-e^{nb}\sigma_n>0\} \nn
&\ge\Tr\rho_n\{\rho_n-e^{nb}\sigma_n>0\}-e^{na}e^{-nb}.
\label{eq:34}
\end{align}
Letting $a=\ud(\whr||\whs)-2\g$ and $b=\ud(\whr||\whs)-\g$ $(\g>0)$, we have
\begin{align}
\liminf_{n\to\infty}\Tr(\rho_n-e^{na}\sigma_n)_+
&\ge\liminf_{n\to\infty}
\left[\Tr\rho_n\{\rho_n-e^{nb}\sigma_n>0\}-e^{-n\g}\right] \\
&=\liminf_{n\to\infty} \Tr\rho_n\{\rho_n-e^{nb}\sigma_n>0\} \\
&\ge 1-\ve,
\label{eq:35}
\end{align}
which implies
\begin{align}
\uc(\ve|\whr||\whs) &\ge a=\ud(\ve|\whr||\whs)-2\g.
\label{eq:36}
\end{align}
Since $\g>0$ can be arbitrary, we have
\begin{align}
\uc(\ve|\whr||\whs) &\ge\ud(\ve|\whr||\whs).
\label{eq:37}
\end{align}
Thus we have (\ref{eq:00}). In the same way, we have (\ref{eq:01}).
\end{proof}

\end{document}